\RequirePackage{amsmath}
\documentclass[runningheads,a4paper]{llncs}
\pdfoutput=1

\usepackage{amssymb}
\usepackage{graphicx}
\usepackage{url}

\usepackage{soul}
\usepackage{amsmath}
\usepackage{complexity}
\usepackage[colorinlistoftodos]{todonotes}

\usepackage{algorithm}
\usepackage{algpseudocode}
\usepackage{tabularx}
\usepackage{todonotes}

\algnewcommand\algorithmicforeach{\textbf{for each}}
\algdef{S}[FOR]{ForEach}[1]{\algorithmicforeach\ #1\ \algorithmicdo}

\newcommand{\qedproof}{\hfill $\square$ \medskip}
\renewenvironment{proof}{\par \noindent {\bf Proof}.}{\qedproof}

\newcommand{\PROBLEM}{CSCN problem}
\newcommand{\STEINER}{ST problem}

\begin{document}

\mainmatter  

\title{Extracting the Groupwise Core Structural Connectivity Network: Bridging Statistical and Graph-Theoretical Approaches}

\titlerunning{Extracting the Groupwise Core Structural Connectivity Network}
\author{\vspace{-2cm}}
\institute{\vspace{-1px}}

\author{Nahuel~Lascano\inst{1,2} 
    \and Guillermo~Gallardo\inst{1}
    \and Rachid~Deriche\inst{1}
    \and Dorian~Mazauric\inst{3}
    \and Demian~Wassermann\inst{1}}
\institute{Athena EPI, Universit\'e C\^ote d'Azur, Inria, France
\and
Computer Science Department, FCEyN, Universidad de Buenos Aires, Argentina
\and
ABS EPI, Universit\'e C\^ote d'Azur, Inria, France}
\authorrunning{N. Lascano et al.} 

\maketitle

\begin{abstract}
Finding the common structural brain connectivity network for a given population is an open problem, crucial for current neuroscience. Recent evidence suggests there's a tightly connected network shared between humans. Obtaining this network will, among many advantages, allow us to focus cognitive and clinical analyses on common connections, thus increasing their statistical power. In turn, knowledge about the common network will facilitate novel analyses to understand the structure-function relationship in the brain.

In this work, we present a new algorithm for computing the core structural connectivity network of a subject sample combining graph theory and statistics. Our algorithm works in accordance with novel evidence on brain topology. We analyze the problem theoretically and prove its complexity. Using 309 subjects, we show its advantages when used as a feature selection for connectivity analysis on populations, outperforming the current approaches.

\end{abstract}

\section{Introduction}

Isolating the common brain connectivity network from a population is a main problem in current neuroscience~\cite{Bullmore2009,Gong2009,Wassermann2016}. Recent evidence suggests that there's a common and densely connected brain connectome across humans~\cite{Bassett2013}. In this work we present a new approach for selecting these common connections, combining recent topological hypotheses~\cite{Bassett2013}  and  current methods~\cite{Gong2009,Wassermann2016}.

Finding the common brain connectome across subjects has the potential to increase our understanding of the relationship between function and structure in the brain. This relationship is one of the main open questions in neuroscience~\cite{Bullmore2009,Donahue2016}. Moreover, knowledge about the most common connections in a population will facilitate clinical and cognitive Diffusion MRI analyses by reducing the number of surveyed connections, increasing the statistical power of those analyses. Finding the common connectome will also allow us to increase our knowledge about the brain structure by comparing core networks across different populations.

We formalize the problem of selecting the common connections combining graph theory and statistics. Then, we prove that the problem is \NP-Hard and propose a polynomial-time algorithm to find approximate solutions. To do this, we develop an exact polynomial-time algorithm for a relaxed version of the problem and prove the algorithm's correctness and complexity.

Currently, the most used algorithm to extract a population's core structural connectivity network (CSNC)~\cite{Gong2009} uses an statistical approach: first, compute a connectivity matrix for each subject; then, analize each connection separately with a hypothesis test, using as null hypothesis that that edge is not present in the population; finally, construct a binary graph with the edges for which the null hypothesis was rejected. The main problem of Gong et al.'s~\cite{Gong2009} algorithm is that the resulting graph can be a set of disconnected subgraphs. Moreover, recent studies have shown that the brain has a \emph{core} network tightly connected and a sparsely connected \emph{outer} one~\cite{Bassett2013}. In other words, this approach ignores the resulting network's topology. Performing statistical analyses in a feature set chosen by hypothesis testing incurs in the double dipping problem~\cite{Kriegeskorte2009}.

A newer approach to solve the CSNC problem, designed by Wassermann et al.~\cite{Wassermann2016}, uses graph theory to get a connected CSCN: first, compute a binary connectivity graph for each subject using a threshold;  for each possible connection compute the ``cost'' of including or excluding it from the common graph by evaluating in how many subjects that connection is present; finally, construct the binary graph with all the edges that is ``cheaper'' to include than to exclude and connect the resulting graph if it's disconnected, using the minimum possible cost. This algorithm guarantees that the resulting graph is connected, but the connection binarization discards significant information for the resulting common network. In other words, it discards information of the probability of each connection being in the brain. This is problematic because the resulting graph may include edges for which tractography assigned a very low existence probability across subjects. Also, the outer part of the brain, the connections which do not result in the core network, should also be sparsely connected~\cite{Bassett2013}, which this algorithm does not enforce.

In this work we propose, for the first time, a polynomial-time algorithm to obtain the CSCN of a population  addressing the issues listed above. Our algorithm combines the recent graph-theoretical approach~\cite{Wassermann2016} with the statistical awareness of the most popular one~\cite{Gong2009}. We start by formalizing the problem, which allow us to prove that it's \NP-Hard. Then, we propose a first algorithm that solves a relaxed version of the problem in an exact way, giving the best possible core graph for our formalization. Then, we adapt it to guarantee a connected result, agreeing with recent evidence on structural connectivity network topology \cite[e.g.]{Bassett2013}. Finally, we validate our approach using 300 subjects from the HCP database and comparing the performance of the networks obtained by our new approach, Wassermann et al.'s~\cite{Wassermann2016} and Gong et al.'s~\cite{Gong2009} predicting connectivity values from handedness in the core network.

\section{Definitions, Problems and Contributions}\label{problem_section}
We want to develop a new algorithm to extract the core structural connectivity network, a problem that implies working with different brains. Thus, the first thing we need to do is unify them into a common connectivity model. This allows us to model all brains with graphs in which each node represents a cortical or sub-cortical region, and each edge represents a white matter connection between two regions. We choose the Desikan parcellation~\cite{Desikan2006} to uniformize the brain cortical and sub-cortical regions across subjects.

To compute the connectivity matrices we use a probabilistic tractography algorithm, which outputs one matrix per subject. The resulting matrices represent the existence probability of a connection across parcels in each subject~\cite{Donahue2016}. As these are symmetric, we interpret the matrices as weighted undirected graphs, sharing the node set across subjects.

Formally, we represent a sample of $N$ brain structural networks by $N$ complete weighted graphs $G_{1} = (V, E, w_1), \ldots, G_N = (V, E, w_N)$ with a common node set $V$. We call $G_1,\ldots,G_N$ the \textit{sample graphs}. Each graph $G_i$ corresponds to a subject. Each vertex $v \in V$ represents a cortical or sub-cortical region. Each edge $e\in E=V\times V$ represents a white matter bundle connecting two regions. Finally, the weight $w_i(e)$ is the connection probability for the edge $e$ in the subject $i$ obtained through tractography:
\begin{equation}\label{all_weights}
w_1(e), w_2(e), \ldots, w_N(e) \in [0,1] ~ \forall e \in E.
\end{equation}
Note that all graphs have the same ordered node set and all of them are complete: an edge weight, or connection probability, $w_i(e)$ of 0 represents an absent connection. Using this formalization we express the general core structural connectivity network problem as follows: find a core graph $G^* = (V^*, E^*)$ densely connected such that $G^*$ keeps the more \emph{relevant} connections $E^*\subseteq E$ in the sample and discards the less \emph{relevant} ones, for some definition of relevance and density. 
For simplicity, once we select $E^*$ we can define $V^*$ as
\begin{equation}
V^* = \{v \in V : \exists u \in V, (u, v) \in E^*\} ~,
\end{equation}
the set of nodes that the edges in $E^*$ cover. Then, we can reduce the problem of finding $G^*$ to find $E^*$ alone.

We want a formalization of relevance that represents the probability that a connection is present across subjects. Thus, we choose to model the group-wise relevance, $w^*(e)$ as the mean existence probability across subjects, factored by the standard deviation of these probabilities. In other words, $w^*(e)$ is the number of standard deviations that the mean existing probability of connection $e$ is larger than $0$. We use $w^*(e)$ as a statistical measure of edge presence across the population. Formally, \begin{equation}\label{weight}
w^*(e) \triangleq \frac{\overline{w(e)}}{s(e)}\text{ where } \overline{w(e)}\triangleq\sum_{i=1}^N\frac{w_i(e)} N,\, {s(e)}\triangleq  \sqrt{\sum_{i=1}^N \frac{\left(\overline{w(e)}-w_i(e)\right)^2}N}.
\end{equation}

Note that $w^*(e)$ is the statistic of a hypothesis z-test which assumes a media of 0 for the population weight of $e$. We choose the z-statistic because of the normal distribution's properties, e.g. linearity, even if other distributions, such as Beta distribution, may be more appropriate for modeling the probability. In any case, note that for the purpose of our contribution $w^*$ can be any function $E \rightarrow \mathbb{R}$ which grows with the relevance of the edges in the sample.

We also want a formalization that represents the density of the core subgraph. We use the relationship between the number of edges and the total statistical relevance $w^*$ that those edges sum:
\begin{equation}\label{alpha}
\alpha(w^*, E^*) \triangleq \frac{\sum_{e \in E^*} w^*(e)}{|E^*|} ~.
\end{equation}
As we want also a sparse outer subgraph, we also define its density:
\begin{equation}\label{beta}
\beta(w^*, E^*) \triangleq \frac{\sum_{e \in E \setminus E^*} w^*(e)}{|E^*|} ~.
\end{equation}

Now we can express our objective informally as: choose $E^*$ such that $\alpha(w^*, E^*)$ (Eq. \ref{alpha}) is large and $\beta(w^*, E^*)$ (Eq. \ref{beta}) small. In accordance to recent evidence on the core network, we also want $G^*$ to be connected. Here, connected means that for every pair of vertices $u, v$ in $V^*$ there is a path of edges in $E^*$ from $u$ to $v$.

Let $\mathcal{E}^{c}$ be the family of sets of edges that induce a connected graph. We now formalize the problem of finding this common graph $G^*$ in two different ways.
\begin{itemize}
\item[$\bullet$] The optimization version consists in computing:
\begin{equation}\label{optimization}
\max_{E^{*} \in \mathcal{E}^{c}} f(w^*, E^*) = \lambda \alpha(w^*, E^*) - (1-\lambda) \beta(w^*, E^*)
\end{equation}
The parameter $\lambda$ (between 0 and 1) can be adjusted to weight the density of the inner and the outer network. Note that if $\lambda = 1$, the solution to (\ref{optimization}) only considers the density of the core network, and if $\lambda = 0$, it only considers the edges excluded of the core network.

\item[$\bullet$] Given $A$ and $B$, the decision version consists in finding $E^* \subseteq \mathcal{E}^{c}$ such that:
\begin{align}
\begin{split}\label{decision}
\alpha(w^*, E^*) &\geq A\\
\beta(w^*, E^*) &\leq B
\end{split}
\end{align}

\end{itemize}

Having formalized the Core Structural Connectivity Network into an optimization and a decision problem, we proceed with one of our main theoretical contributions: proving that the problem is \NP-Complete.

\subsection{CSCN Problem's NP-Completeness}
We have formalized the problem of the Core Structural Connectivity Network taking into account the density and connectedness of the core subgraph and the sparsity of the outer one. We will now prove that, with this formalization, the problem is \NP-Complete.

\begin{definition}[Core Structural Connectivity Network problem]
Given $G_1 = (V, E, w_1), G_2 = (V, E, w_2), \ldots, G_N = (V, E, w_N)$ weighted graphs (the \emph{sample graphs}) with a common node set, a complete edges set ($E = V \times V$) and $w_1(e), w_2(e), \ldots, w_N(e) \in \mathbb{R}_{\geq0} ~ \forall e \in E$ weights of their edges, and given $A, B$ real numbers, find $G^* = (V^*, E^*)$ connected graph (the \emph{core graph}) such that
$$\alpha(w^*, E^*) \geq A$$
$$\beta(w^*, E^*) \leq B$$
for $\alpha$ and $\beta$ as defined in Eq.~(\ref{alpha}) and Eq~.(\ref{beta}).
\end{definition}

Here we prove that the \emph{Core Structural Connectivity Network problem}, called \PROBLEM, is NP-complete. In our reduction, we use the \emph{Steiner Tree problem}~\cite{Garey1979}, called \STEINER~in the
following. Given an edge-weighted graph $G' = (V',E',w)$, a subset
$S \subseteq V'$ of nodes, and a real $k \geq 0$, \STEINER~consists in
determining if there exists a connected subgraph $H$ such that $S \subseteq V(H)$ and $\sum_{e \in E(H)}{w(e)} \leq k$.
The decision version of \STEINER~is NP-complete even if all weights are equal~\cite{Garey1979}.

\smallskip

\textbf{Instance of \STEINER.}
Consider any edge-weighted graph $G' = (V',E',w)$ such that $w(e) =
\frac{1}{2}$ for every $e \in E'$.
Given $k \geq 0$, \STEINER~consists in determining if there exists a
connected subgraph $H$ such that $S \subseteq V(H)$ and $|E(S)| \leq
2k$.
Without loss of generality, we assume that $|E'| \geq 2k$ and that
$G'$ is connected.

\smallskip

\textbf{Reduction.}
We construct the instance of \PROBLEM~as follows.
Let $s = |S|$ and let $t \geq 1$ be any positive integer.
Let $G = (V,E,w^{*})$ defined as follows.
Let $V = V' \cup \{v_{i,j} \mid 1 \leq i \leq s, 1 \leq j \leq t\}$
and $E = V \times V$.
Let $S = \{u_{1}, \ldots, u_{s}\}$.
For every $i,j$, $1 \leq i \leq s$, $1 \leq j \leq t$,
$w^{*}_{v_{i,j},u_{i}} = 1$ and $w^{*}_{v_{i,j},u} = 0$ for every $u
\in V \setminus \{u_{i}\}$.
Furthermore, for every $e \in E'$, set $w^*(e) = w(e) =
\frac{1}{2}$, and for every $u,u' \in V'$ such that $\{u,u'\} \notin
E'$, then set $w^*_{e} = 0$.
Finally, we set $A = \frac{s.t+k}{s.t+2k}$ and $B =
\frac{\frac{1}{2}(|E'|-2k)}{s.t+2k}$.

\begin{lemma}
\label{lem:beta}
If $|E^*| < s.t+2k$, then any solution for \PROBLEM~is not
admissible because $\beta(w^{*},E^*) > B$.
\end{lemma}

\begin{proof}
Suppose that $|E^*| < s.t+2k$.
In order to minimize $\sum_{e \in E \setminus E^*} w^*(e)$,
$E^*$ must contain $\{\{v_{i,j},u_{i}\} \mid 1 \leq i \leq s, 1 \leq
j \leq t\}$ if $|E^*| \geq s.t$.
(Otherwise we select a subset of this set of edges.)
Indeed, by construction of $G$, we have $w^{*}_{v_{i,j},u_{i}} = 1$
for every $i,j$, $1 \leq i \leq s$, $1 \leq j \leq t$.
Then, if $|E^*| - s.t > 0$, $E^*$ must contain $|E^*| - s.t$
edges of $E'$, that is edges of $E$ of weight $\frac{1}{2}$ each.
Recall that there are exactly $s.t$ edges of weight $1$, and the other
edges have weight $0$ or $\frac{1}{2}$.

There are two cases.
First, suppose that $|E^*| \geq s.t$.
We get that $\sum_{e \in E \setminus E^*} w^*(e) =
\frac{1}{2}(|E'| - (|E^*| - s.t))$.
Since $|E^*| - s.t < 2k$, then we get that $\sum_{e \in E \setminus E^*} w^*(e) =
\frac{1}{2}(|E'| - (|E^*| - s.t)) > \frac{1}{2}(|E'| - 2k)$.
Furthermore, since $|E^*| < s.t+2k$, we get that 
$\frac{\frac{1}{2}(|E'| - (|E^*| - s.t))}{|E^*|} >
\frac{\frac{1}{2}(|E'|-2k)}{s.t+2k}$.
Thus, we proved that $\beta(w^{*},E^*) = \frac{\sum_{e \in E \setminus E^*}
  w^*(e)}{|E^*|} > B$.

Second, suppose that $|E^*| < s.t$.
We get that $\sum_{e \in E \setminus E^*} w^*(e) = s.t - |E^*|
+ \frac{|E'|}{2}$.
Since $s.t - |E^*| + \frac{|E'|}{2} > \frac{1}{2}(|E'| - (|E^*| -
s.t))$, we obtain the result by the arguments described for the first
case.

Finally, if $|E^*| < s.t+2k$, then there is no admissible solution
for \PROBLEM.
\end{proof}

\begin{lemma}
\label{lem:alpha}
If $|E^*| > s.t+2k$, then any solution for \PROBLEM~is not
admissible because $\alpha(w^{*},E^*) < A$.
\end{lemma}

\begin{proof}
  Suppose that $|E^*| > s.t+2k$.
  In order to maximize $\sum_{e \in E^*} w^*(e)$,
$E^*$ must contain $\{\{v_{i,j},u_{i}\} \mid 1 \leq i \leq s, 1 \leq
  j \leq t\}$ and $|E^*| - s.t$ edges of $E'$.
Indeed, by construction of $G$, we have $w^{*}_{v_{i,j},u_{i}} = 1$
for every $i,j$, $1 \leq i \leq s$, $1 \leq j \leq t$.
Furthermore, $w^*(e) = \frac{1}{2}$ for every $e \in E'$.
Recall that there are exactly $s.t$ edges of weight $1$, and the other
edges have weight $0$ or $\frac{1}{2}$.
We get that $\alpha(w^{*},E^*) = \frac{s.t + \frac{1}{2}(|E^*| -
  s.t)}{|E^*|} < \frac{s.t+k}{s.t+2k} = A$.
Indeed, the average weight is lower when there are more edges of
weight $\frac{1}{2}$ (the number of edges of weight $1$ is the same in
both ratios).

Finally, if $|E^*| > s.t+2k$, then there is no admissible solution
for \PROBLEM.
\end{proof}

By Lemma~\ref{lem:beta} and Lemma~\ref{lem:alpha}, we get the
following corollary.

\begin{corollary}
\label{cor:alpha-beta}
Any solution for \PROBLEM~is such that $|E^*| = s.t + 2k$.
\end{corollary}

We prove in Lemma~\ref{lem:droite-gauche} and in
Lemma~\ref{lem:gauche-droite} that there is an admissible solution for
\PROBLEM~if and only if there is an admissible solution for \STEINER.

\begin{lemma}
\label{lem:droite-gauche}
If there is an admissible solution for \STEINER, then there is an
admissible solution for \PROBLEM.
\end{lemma}

\begin{proof}
Suppose there is an admissible solution for \STEINER.
We prove that there is an admissible solution for \PROBLEM.
Let $H$ be a connected subgraph such that $S
\subseteq V(H)$ and $\sum_{e \in E(H)}{w(e)} = \frac{1}{2}|E(|H)|
\leq k$.
If $|E(H)| = 2k$, then set $E^* = E(H) \cup \{\{v_{i,j},u_{i}\} \mid 1 \leq i \leq s, 1 \leq j \leq
t\}$.
If $|E(H)| < 2k$, then set $E^* = E(H) \cup \{\{v_{i,j},u_{i}\} \mid 1 \leq i \leq s, 1 \leq j \leq
t\} \cup F$, where $F \subseteq E'$ such that $F \neq E(H) =
\emptyset$, $w^*(e) = \frac{1}{2}$ for every $e \in F$, and such
that the graph induced by $E(H) \cup F$ is connected.
The last condition comes from Corollary~\ref{cor:alpha-beta} in order
to get the right number of edges in $E^*$.
This condition is always possible to satisfy because $G'$ is
connected.

The graph induced by $E^*$ is connected.
Indeed, $H$ is an admissible solution for \STEINER, $E(H) \cup F$ is
connected by construction, and $\{\{v_{i,j},u_{i}\} \mid 1 \leq j \leq
t\}$ is a set of edges all adjacent to $u_{i} \in S$ for every $i$, $1
\leq i \leq s$.

Furthermore, we get
$$\alpha(w^{*},E^*) = \frac{\sum_{e \in E^*}
  w^*(e)}{|E^*|} = \frac{s.t+k}{s.t+2k} = A$$
and
$$\beta(w^{*},E^*) = \frac{\sum_{e \in E \setminus E^*}
  w^*(e)}{|E^*|} = \frac{\frac{1}{2}(|E'|-2k)}{s.t+2k} = B.$$

Finally, we proved that there is an admissible solution for \PROBLEM.
\end{proof}

\begin{lemma}
\label{lem:gauche-droite}
If there is an admissible solution for \PROBLEM, then there is an
admissible solution for \STEINER.
\end{lemma}

\begin{proof}
Suppose there is an admissible solution for \PROBLEM.
We prove that there is an admissible solution for \STEINER.
Let $E^* \subseteq E$ be such that the graph induced by $E^*$ is
connected, and such that $\alpha(w^{*},E^*) \geq
\frac{s.t+k}{s.t+2k} = A$ and $\beta(w^{*},E^*) \leq =
\frac{\frac{1}{2}(|E'|-2k)}{s.t+2k} = B$.

We first prove that $\{\{v_{i,j},u_{i}\} \mid 1 \leq i \leq s, 1 \leq j \leq
t\} \subseteq E^*$.
By Corollary~\ref{cor:alpha-beta}, we know that $|E^*| = s.t + 2k$.
Thus, it necessarily means that $\sum_{e \in E^*} w^*(e) \geq s.t
+ k$.
By construction of $G$, the set of edges of weight $1$ is $\{\{v_{i,j},u_{i}\} \mid 1 \leq i \leq s, 1 \leq j \leq
t\}$, that is $\{e \in E \mid w^*(e) > \frac{1}{2}\} =
\{\{v_{i,j},u_{i}\} \mid 1 \leq i \leq s, 1 \leq j \leq t\}$.
We get that $\{\{v_{i,j},u_{i}\} \mid 1 \leq i \leq s, 1 \leq j \leq
t\} \subseteq E^*$ because, otherwise, we would have $\sum_{e \in E^*} w^*(e) < s.t
+ k$.

Furthermore, every $e \in E^* \cap E'$ is such that
$w^*(e) = \frac{1}{2}$.
Indeed, otherwise, we would have $\sum_{e \in E^*} w^*(e) < s.t
+ k$.

Finally, since $E^*$ is an admissible solution for \PROBLEM,
then it means that the graph induced by the set of edges $E^* \cap
E'$ is connected and is such that for every $u_{i}$, $1 \leq i \leq
s$, then there is an edge $e \in E^* \cap E'$ that is adjacent to
$u_{i}$.
By the previous remark, every edge in $E^* \cap E'$ has weight
$\frac{1}{2}$.
Thus, it means that there is $E^* \cap E'$ is an admissible solution
for \STEINER~considering the graph $G'$.
Indeed $|E^* \cap E'| = 2k$ and so $\sum_{e \in E^* \cap E'}w(e)
= k$.
\end{proof}

We are now able to prove the NP-completeness of \PROBLEM.

\begin{theorem}
\label{the:NP-complete-with-Steiner}
\PROBLEM~is NP-complete.
\end{theorem}
\begin{proof}
The reduction is clearly polynomial.
Furthermore, Lemma~\ref{lem:droite-gauche} and
Lemma~\ref{lem:gauche-droite} prove the equivalence between
\PROBLEM~and \STEINER.
Since the decision version of \STEINER~is NP-complete even if all weights are
equal~\cite{Garey1979}, then we obtain the NP-completeness of the
decision version of \PROBLEM.
\end{proof}



In Theorem~\ref{the:NP-complete-with-Steiner} we have proved that \PROBLEM~is \NP-complete. Hence, to be able to solve it in reasonable time we need a relaxation to make it tractable or an approximate algorithm for the complete version. In this article we will propose both.

\subsection{Relaxation of the CSCN problem}\label{algorithm_section}
We proved in the previous section that the connectivity constraint is the main reason of the difficulty of the problem. Without it, it becomes tractable. So we solve, in this section, a relaxed version of problem without the connectivity constraint. Then, we use this solution to approximate the full problem.

\begin{theorem} The decision version of \PROBLEM~without the connectivity constraint is in \P.
\end{theorem}

\begin{algorithm}
\caption{Maximum edges}
\label{alg_max_edges}
\begin{algorithmic}
\State Compute $w^*(e)$ for each $e \in E$
\State \Call{Sort}{$E$} \Comment{sorts edges by $w^*$ non-increasingly}
\ForEach {$e \in E$}
	\State $E^* \gets E^* \cup {e}$ 
    \If {$\alpha(w^*, E^*) > A$ and $\beta(w^*, E^*) < B$}
		\State \Return $True$
    \EndIf
\EndFor
\State \Return $False$
\end{algorithmic}
\end{algorithm}

\begin{proof}
Algorithm \ref{alg_max_edges}, in each step $i$, defines $E^*$ as the $i$ maximum weighted edges and tries to use that to fulfill the constraints.

Assume that there exists an $E^*$ that fulfills the constraints. There are two cases: 1) $E^*$ has the $|E^*|$ maximum weighted edges, 2) there are $e_j \in E^*$, $e_k \in E \setminus E^*$ such that $w^*(e_k) \geq w^*(e_j)$.
In 1), Algorithm \ref{alg_max_edges} will find $E^*$.
In 2), let $E' = (E^* \cup \{e_k\}) \setminus \{e_j\}$ another subset of $E$. Then
$$\alpha(w^*, E') = \frac{\sum_{e \in E'} w^*(e)}{|E'|} = \frac{\sum_{e \in E'} w^*(e)}{|E^*|} \geq \frac{\sum_{e \in E^*} w^*(e)}{|E^*|} = \alpha(w^*, E^*) \geq A$$
because the edges in $E^*$ are the same as the ones in $E'$ except from one that has a larger weight. For the same reason,
$$\beta(w^*, E') = \frac{\sum_{e \in E \setminus E'} w^*(e)}{|E'|} = \frac{\sum_{e \in E \setminus E'} w^*(e)}{|E^*|} \leq \frac{\sum_{e \in E \setminus E^*} w^*(e)}{|E^*|} = \beta(w^*, E^*) \leq B.$$
Thus, we found a new subset of $E$ that stills fulfills the constraints. We can do the same process with $E'$ (replace an edge with another one of larger weight) iteratively, always getting subsets that fulfills the constraints, until we cannot do this anymore. At that point we will have a subset that has only the maximum $|E^*|$ edges and fulfills the constraints. Thus, algorithm \ref{alg_max_edges} will find this subset.

We now need to prove algorithm \ref{alg_max_edges} runs in polynomial time in the size of $|E|$. The first operation, computing $w^*(e)$ for each $e$, implies computing the mean and standard deviation for each edge across the population, which can be done in $\mathcal{O}(N)$ per edge (where $N$ is the size of the population). This is $\mathcal{O}(N*|E|)$ for all the edges. The second step, sorting, can be done in $\mathcal{O}(|E|\log{|E|})$.

The main loop runs at most $|E|$ times, and in each loop it adds an edge to $E^*$, computes $\alpha$ and $\beta$ and performs two comparisons. The comparisons can be done in constant time, as the addition to $E^*$ if we use a linked list of edges to represent it. To compute $\alpha$ and $\beta$ it is needed to iterate once again $E$ (the part in $E^*$ for $\alpha$, the part in $E \setminus E^*$ for $\beta$) adding the weights together and then performing two divisions. This can be done in linear time in the size of $E$, and even quicker (constant time) if we optimize it by keeping the values of $\alpha$ and $\beta$ across loops and updating them with the weight of the edge that changed sets. 

Then, algorithm \ref{alg_max_edges} solves the \PROBLEM~in $\mathcal{O}(max(|E|^2, |E|*N)$ or in $\mathcal{O}(|E|*N)$ if a little optimization is used.

\end{proof}

\subsection{Heuristic approach}\label{heuristic}
In Section~\ref{algorithm_section} we developed  Algorithm~\ref{alg_max_edges} to solve the problem of finding the Core Structural Connectivity Network in polynomial time. However, this algorithm does not guarantee a connected result. We solve the original problem, presented in Section~\ref{problem_section}, by first applying Algorithm~\ref{alg_max_edges} and then modifying the resulting core graph $G^*$ to guarantee its connectedness. This results in an approximate solution for the full problem computable in polynomial time.

To extend $G^*$ into a connected graph we add the necessary edges while decreasing the minimum possible the objective function $f$ defined in Eq.~\ref{optimization}. For this, we use the same approach that Wassermann et al.~\cite{Wassermann2016}. Namely, we make a multigraph $G_{cc}$ with the connected components of $G^*$ as nodes, complete it with all the possible edges between those connected components, and run a Maximum Spanning Tree algorithm. This selects the edges needed to produce a connected subgraph with the maximum possible weight. For the full details, see Wassermann et al.~\cite{Wassermann2016}. This way we get a connected subgraph close to the best possible subgraph, which we obtained using Algorithm~\ref{alg_max_edges}.

\section{Experiments and Results}
\begin{figure}\center%
  \includegraphics[height=.15\textheight]{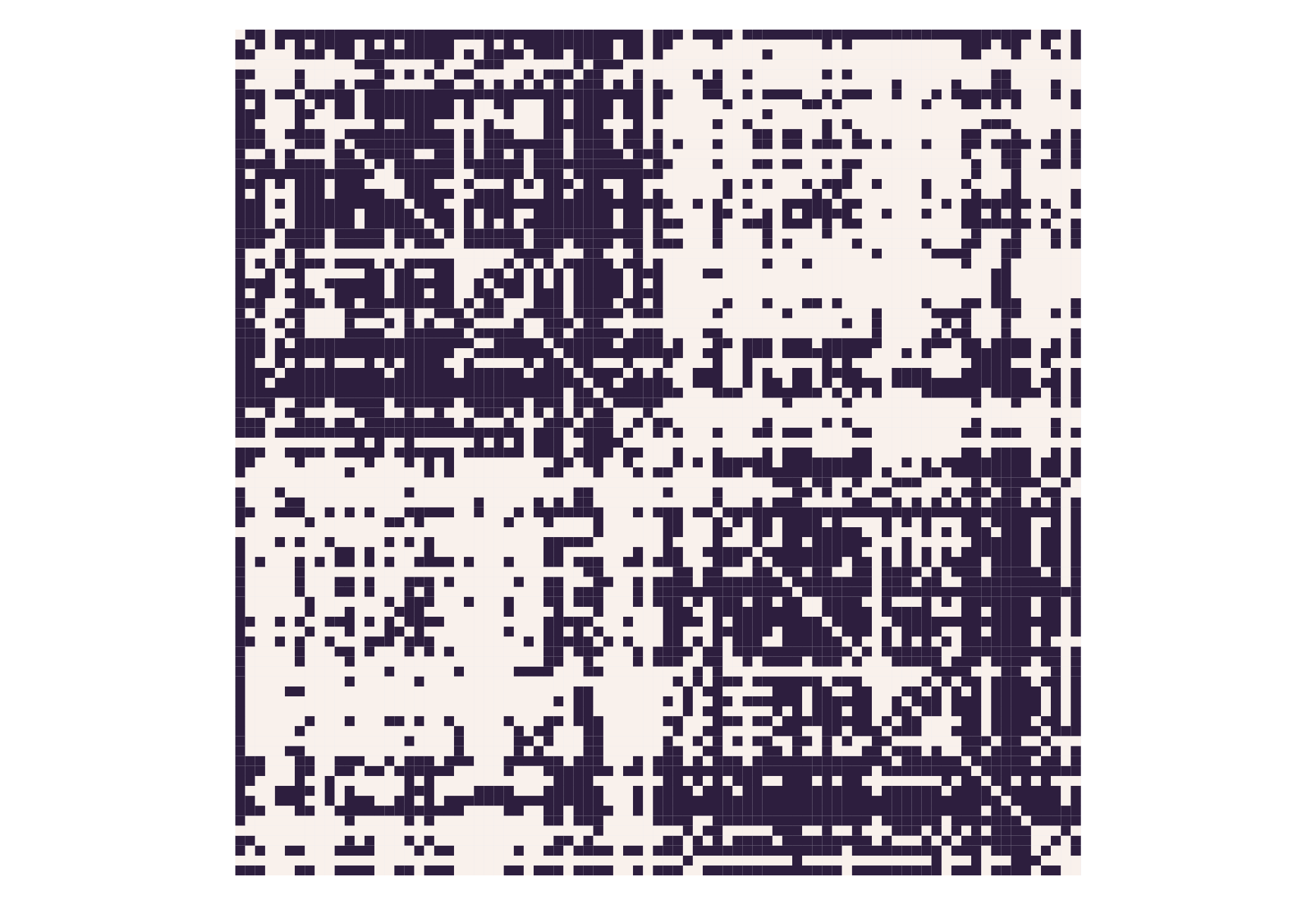}
  \includegraphics[height=.15\textheight]{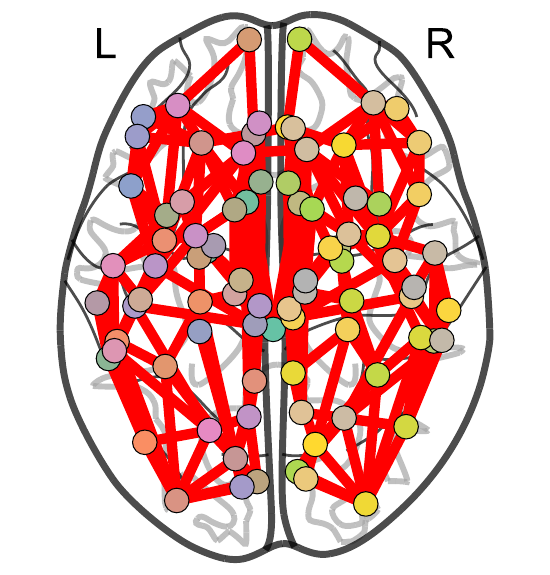}
  \includegraphics[height=.15\textheight]{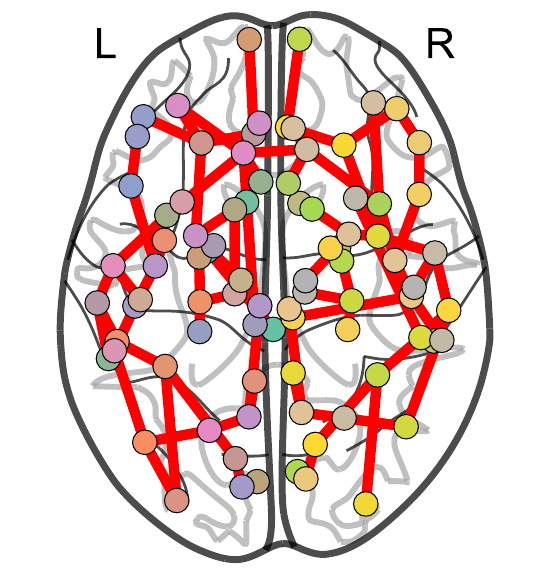}
  \caption{Core structural connectivity network computed by our approach. On the left, we show the adjacency matrix for $\lambda = 0.5$, where 48.19\% of the connections were included in the CSCN. In the central and right panels, we show the resulting CSCN for $\lambda=0.9$ and $0.99$ respectively. The percentage of included connections in the CSCN is 5.99\% and 1.27\% respectively.\label{fig:connectome}}
  
\end{figure}

We formalized the CSCN problem in section \ref{problem_section} and designed an algorithm to solve it in section \ref{heuristic}. Now we will asses the performance of our method. For this, we compare it with the most used~\cite{Gong2009} and with the recent one~\cite{Wassermann2016} in the task of connectivity prediction performance.

We use a subset of the HCP500 dataset~\cite{Sotiropoulos2013}: all subjects aged 21-40 with complete dMRI protocol, totaling 309. We compute the weighted connectivity matrices between the cortical regions defined by the Desikan atlas~\cite{Desikan2006} as done by Sotiropoulos et al.~\cite{Sotiropoulos2013}.
Examples of CSCN exctracted with our algorithm at different $\lambda$ levels are shown in Fig.~\ref{fig:connectome}, which was generated using Nilearn~\cite{Abraham2014}.

\subsection{Consistency of the Extracted Graph}
To compare the stability across different algorithms for CSCN, we use an analysis based on Wassermann et al.~\cite{Wassermann2016}: we randomly take $500$ subsets of $100$ subjects each and computed the core graphs for all subsets. We then compute the number of \emph{unstable connections}: connections that present in at least one core graph but not in all of them. Finally, in Table~\ref{table:number_of_features} we report each algorithm's \emph{stability}: 

$$\text{stability of the algorithm} \triangleq 1 - \frac{\#\{\text{unstable connections}\}}{\#\{\text{total connections}\}}~.$$

This measure quantifies the CSCN consistency across subsamples. Due to the homogeneity of our sample, we expect the CSCNs obtained by an algorithm to be similar across subsamples. Hence, a stabler algorithm is preferable.


\subsection{Predicting Handedness-Specific Connectivity}

We evaluate performance of the methods by using the generated core graphs as a feature selection for handedness specific connectivity. We use a nested Leave-$\tfrac 1 3$-Out procedure: the outer loop performs model selection on $\tfrac 1 3$ of the subjects using the core graph algorithm and the inner loop performs model fitting and prediction using the selected features.

\begin{table}
\centering

\caption{Stability of the algorithms and amount of features selected by linear regression in the core graph relating the weights with handedness. Our procedure gets more features selected than Gong et al.~\cite{Gong2009} and Wassermann et al.~\cite{Wassermann2016}, showing better statistical power. It's also more stable than Gong et al., showing improved consistency. \label{table:number_of_features}}
\begin{tabular}{ l c c c }
	\hline
    Algorithm & Features (mean) & Features (std) & Stability \\
    \hline
    Gong et al. 2009& 0.066 & 0.256 & 0.364 \\
    Wassermann et al. 2016 & 0.415 & 0.723 & 0.644 \\
    Our approach ($\lambda = 0.50$) & 1.042 & 1.269 & 0.528 \\
\end{tabular}
\end{table}

\begin{figure}\center%
  \includegraphics[height=.3\textheight]{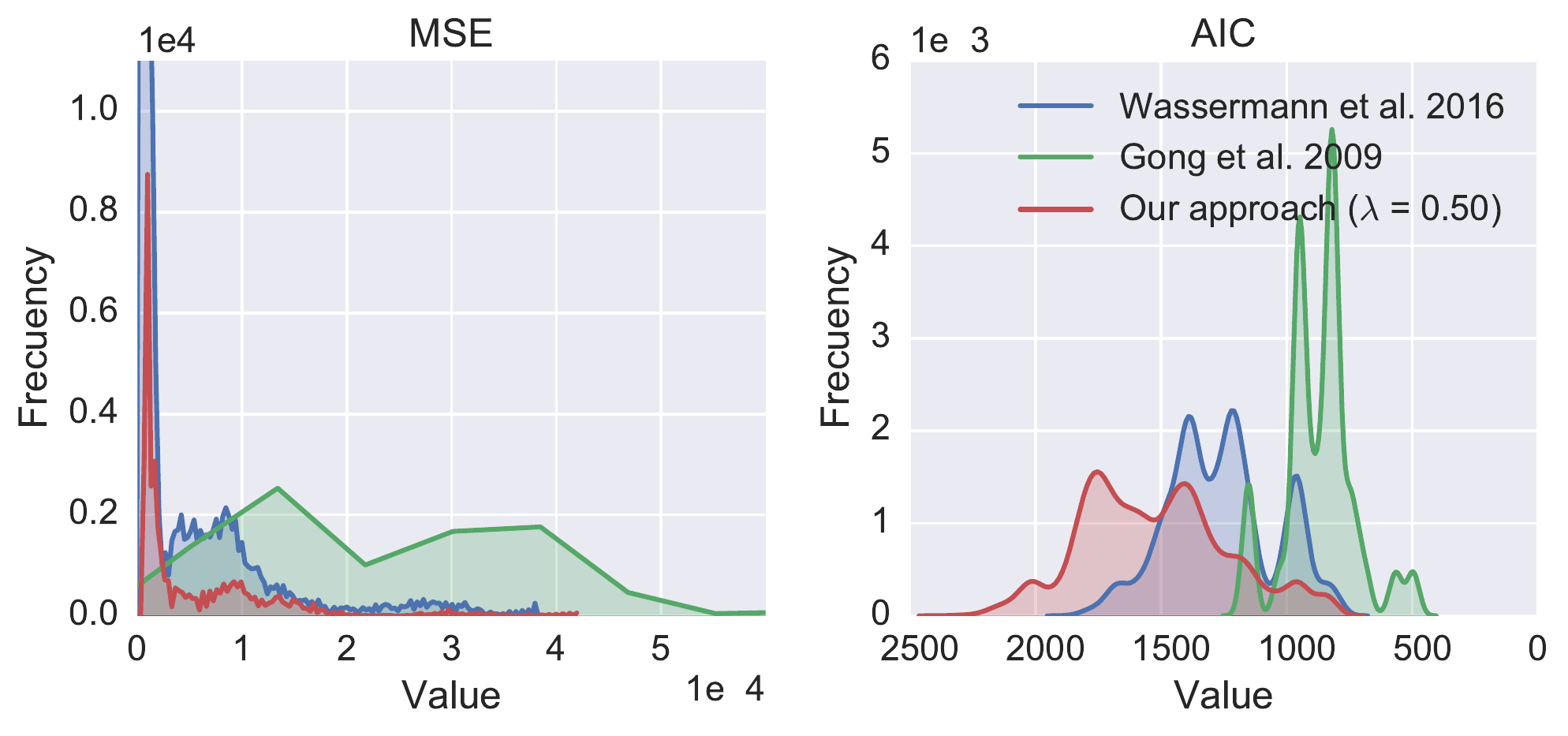}
  \caption{Performance of core network as feature selection for a linear model for handedness specific connectivity. We evaluate model prediction (left) and fit (right) for Gong et al.~\cite{Gong2009} in green, Wassermann et al.~\cite{Wassermann2016} in blue and ours, in red. We show the histograms from our nested Leave-$\tfrac 1 3$-Out experiment. In both measures, our approach has more frequent lower values than Gong et al., showing a better performance.\label{fig:prediction_performance}}
\end{figure}

Specifically, we first take $\tfrac 1 3$ subjects randomly and compute the core graph for those subjects using the three different algorithms. Then we add the weights for the selected edges for each subject, and select the features $F$ that are more determinant of handedness using a linear least-squares regression and the Bonferroni correction for multiple hypotheses. This experiment is repeated 500 times. We quantify the amount of features that are selected after each procedure, which indicates how useful is the core graph algorithm for selecting the edges related to handedness. 
We show the results in Table~\ref{table:number_of_features}.

To evaluate the prediction, we randomly take $\tfrac 1 2$ of the remaining subjects and fit a linear model on the features $F$ to predict connectivity weights using the handedness of each subject. Finally, we predict the values of the features $F$ from the handedness in the subjects left out. We quantify the quality of the linear model fitting Akaike Information Criterion (AIC) and of the prediction performance with the mean squared error (MSE) of the prediction. For both measures a lower value indicates better performance. The outer loop is performed 500 times and the inner loop 100 times per outer loop, which totals 50,000 experiments. We show the results of this experiments in Fig.~\ref{fig:prediction_performance}.

\section{Discussion and Conclusion}
We presented for the first time a polynomial algorithm to extract the core structural connectivity network of a population combining a graph-theoretical approach with statistic relevance of the connections, observing the recent evidence of the structural network topology.

Our results show that our algorithm outperforms, in the prediction experiment, the most used technique~\cite{Gong2009} as well as latest approaches~\cite{Wassermann2016}. In Table~\ref{table:number_of_features} we can see that our algorithm preserves, in average, more connections correlated with the handedness of the subjects. We can also see that despite being less stable than Wassermann et al.'s it is stabler than Gong et al.'s. Finally, Fig.~\ref{fig:prediction_performance} shows that, in the handedness prediction experiment, our method outperforms  Gong et al.'s and Wassermann et al's: the number of cases with lower AIC and MSE is larger in our case. Hence, our CSCN is better as linear model relating connectivity with handedness in terms of model fitting and prediction.

In terms of theoretical contributions, we formalized the problem, proved its difficulty and gave a novel algorithm for dealing with it. We then validated our approach by showing its power as feature selector for getting connections related to handedness with 300 real subjects' data. The experiment shows our method performs better than the currently available. Moreover, our method avoids the double dipping problem by not choosing the feature set with hypothesis testing.

{\small \textbf{Acknowledgements} Authors acknowledge funding from ERC Advanced Grant agreement No 694665 : CoBCoM - Computational Brain 
Connectivity Mapping}

\vspace{-.5cm}
\bibliographystyle{splncs03}
\bibliography{library}

\end{document}